\pdfoutput=1

\documentclass[letterpaper, 10 pt, conference, twocolumn,final]{ieeeconf}

\IEEEoverridecommandlockouts                              

\usepackage[mode=buildnew]{standalone}

\RequirePackage[l2tabu, orthodox]{nag}

\usepackage[american]{babel}
\usepackage[
			hmargin=2cm,
			vmargin=2cm,
			nomarginpar]
			{geometry}

\usepackage{amssymb,amsthm}
\usepackage[fleqn]{amsmath}
\usepackage[notref,notcite]{showkeys}
\usepackage{xspace} 
\usepackage[dvipsnames]{xcolor}
\usepackage{todonotes}
\usepackage{dsfont}
\usepackage[tracking=true]{microtype}
\DisableLigatures[f]{encoding = *}

\definecolor{dred}{RGB}{220,0,0}
\definecolor{halfgray}{gray}{0.55}
\definecolor{webgreen}{rgb}{0,.5,0}
\definecolor{webbrown}{rgb}{.6,0,0}

\usepackage{tabularx} 
\usepackage{placeins}
\usepackage{graphicx}

\usepackage{bm}

\newcommand{\ie}{i.e.,~}
\newcommand{\eg}{e.g.,~}

\newcommand{\E}{\mathbb E}

\newcommand{\R}{\mathbb R}

\renewcommand{\a}{\alpha}

\newcommand{\s}{\sigma}
\newcommand{\g}{\gamma}

\newcommand{\xx}{\mathbf{x}}
\newcommand{\yy}{\mathbf{y}}
\newcommand{\ff}{\mathbf{f}}
\newcommand{\FF}{\mathbf{F}}

\newcommand{\tr}[1]{\mathrm{tr}(#1)}
\newcommand{\prm}{P_{\mathrm{s,1}}}
\newcommand{\prmt}{P_{\mathrm{s,2}}}
\newcommand{\ls}{L^+_{\mathrm{s,1}}}
\newcommand{\lst}{L^+_{\mathrm{s,2}}}
\newcommand{\nv}{n}
\renewcommand{\ne}{m}
\newcommand{\D}{\Delta}

\newtheorem{thm}{Theorem}[section]

\newtheorem{lem}[thm]{Lemma}
\newtheorem{prop}[thm]{Proposition}
\theoremstyle{definition}

\title{\LARGE \bf Minimizing heat loss in DC networks using batteries}

\author{Alessandro Zocca \and Bert Zwart 
\thanks{This work is supported by an NWO VICI grant.}
\thanks{A. Zocca and B. Zwart are with the Centrum Wiskunde \& Informatica, Science Park 123, 1098 XG, Amsterdam, The Netherlands. Email: {\tt\small zocca@cwi.nl}}}%

\begin{document}
\maketitle
\thispagestyle{empty}
\pagestyle{empty}


\begin{abstract}

Electricity transmission networks dissipate a non-negligible fraction of the power they transport due to the heat loss in the transmission lines. In this work we explore how the transport of energy can be more efficient by adding to the network multiple batteries that can coordinate their operations. Such batteries can both charge using the current excess in the network or discharge to meet the network current demand. Either way, the presence of batteries in the network can be leveraged to mitigate the intrinsic uncertainty in the power generation and demand and, hence, transport the energy more efficiently through the network.
We consider a resistive DC network with stochastic external current injections or consumptions and show how the expected total heat loss depends on the network structure and on the batteries operations. Furthermore, in the case where the external currents are modeled by Ornstein-Uhlenbeck processes, we derive the dynamical optimal control for the batteries over a finite time interval.

\end{abstract}


\section{Introduction}
\label{sec1}
The rise of renewable energy has tremendous impact on the redesign of electricity transmission and distribution networks. Renewable energy sources are significantly more variable than traditional energy sources, and are already responsible for 80 percent of the bottlenecks in the European transmission grid~\cite{Y15}. On the other hand, the availability of new technologies creates new opportunities. In particular, the incorporation of storage devices in the network has been suggested to mitigate variability, both from the supply and demand side.

The motivation of this work comes from the observation that energy networks, unlike communication and road traffic networks, lose a nonnegligable amount of the entity they are trying to transport, namely energy. One reason, which is the focal point of this paper, is heat loss. The central question we aim to address in this paper is the following: What is the impact of renewable energy on the amount of heat loss in an energy network, and how can storage devices such as batteries help mitigate this loss? We pursue these questions also in view of the fact that the total heat loss is a good candidate as a global performance metric of the efficiency of transport in a power network.

In a smart-grids context, storage devices have been analyzed in a variety of settings. In a macroscopic setting, storage models have been suggested to deal with uncertainties in wind generation that play a role on longer time scales; in these models the physical network is not modeled explicitly, see for example~\cite{BGK12,GTLB12}. Storage can also be used as tactical tool, exploiting temporal price differences~\cite{CZ15,CGZ14}.


The effect of storage and the optimal storage placement has been studied in the framework of optimal power flows, see \eg \cite{BGTC12,CLTX10}, but mostly assuming non-stochastic generation and demand. In \cite{BCZ16}, it is illustrated by simulation techniques how optimal storage placement can increase network reliability.

Several papers in the literature also investigate optimal policies for storage management, both at grid operator level~\cite{KHT11} as well as at consumer level~\cite{vdVHMS13}.  Other recent papers on control problems in energy systems focus on demand-side management, see e.g. \cite{MBBCE15}. Related papers on storage modeling in energy systems are~\cite{KMTZ11,SEL11,WTLJ10,ZZTL16}.

%

The above-mentioned papers do not model the transport network explicitly, are restricted to simulations, or focus on deterministic scenarios. Our paper proposes a probabilistic analysis that takes the network structure into account. For reasons of mathematical tractability, we model our energy network as DC network, whose underlying equations are the same as in the widely used linear approximation for AC networks, see \eg \cite{CKvM16}. In a resistive network the total amount of heat loss can be described through Thompsons principle~\cite{Whittle2007}, and follows from the analysis of a weighted graph with random sources and sinks.

In general, storage devices can fill up or empty. In this work, we focus on short time scales in the network, up to a few minutes, and make the assumption that such boundary effects do not play a role. The time-scale assumption also allows us to model the input flows as Ornstein-Uhlenbeck processes; for a physical motivation we refer to~\cite{BCH14}.

In the case of a single battery, these assumptions reduce the problem to the total heat loss in a DC network having random sources and sinks with a single slack node. We derive an explicit expression for the expected total heat loss that gives insight in how the optimal location of such a battery depends on the network properties as well as on the mean and variance of the external currents.

For several reasons, it makes sense to use multiple batteries in a network, the most obvious one is, as we show here, that the total amount of heat loss can be reduced significantly. Considering the prototypical case of two batteries, we show in a static setting how the expected heat loss crucially depends on the way the batteries share the load.

We then investigate a control problem in which we allow batteries to share the load dynamically, minimizing the total expected heat loss, subject to a constraint which makes sure that the load does not oscillate from one battery to the other. The sharing between batteries is a deterministic control problem, and as we show, leads to an explicit solution if the input is stationary.

Our paper is organized as follows. The model is described in Section~\ref{sec2}. Sections~\ref{sec3} and~\ref{sec4} describe the cases of one and two batteries, respectively, in a static context. We examine a dynamic load sharing problem for two batteries in Section~\ref{sec5} and present our conclusion in Section~\ref{sec6}.


\section{Model description}
\label{sec2}
In this paper we model an energy network as a weighted graph $(G,w)$, where the graph $G$ is a simple undirected graph $G=(V,E)$ with $|V| =n$ nodes and $|E|=\ne$ edges and $w \in \R_+^m$ is the collection of edge weights.

In the energy network terminology, $(G,w)$ is a \textit{resistive electrical network} with $n$ \textit{buses} and $m$ \textit{transmission lines} and the quantity $w_{i,j}$ is the \textit{conductance} of the transmission line connecting buses $i$ and $j$.

Without loss of generality, we will henceforth assume that $G$ is a connected graph. The \textit{weighted Laplacian matrix} of the graph $G$ is the $n \times n$ matrix $L$ defined as
\[
    L_{i,j} :=
    \begin{cases}
    	\sum_{k \neq i} w_{i,k} 		& \text{ if } i=j,\\
	- w_{i,j}  								& \text{ if } i \neq j.
	\end{cases}
\]
The matrix $L$ satisfies the following identity
\begin{equation}
\label{eq:LADA}
	 L =A^T D A,
\end{equation}
where $A \in \R^{\ne\times \nv}$ is \textit{oriented edge-vertex incidence matrix} (after having chosen any orientation for the edges), namely
\[
  A_{l,k} :=
  \begin{cases}
      0 & \text{ if } l=(i,j), \, k\notin \{i,j\},\\
      1 &\text{ if } l=(i,j), \, k=i,\\
      -1 & \text{ if } l=(i,j), \, k=j,
   \end{cases}
\]
and $D\in\mathbb{R}^{\ne\times \ne}$ is the diagonal matrix with the edge conductances as diagonal elements, \ie $D_{l,l}:=w_{i,j}$, if $l=1,\dots,m$ is the index corresponding to the edge $(i,j)$.

Let $L^+ \in \R^{\nv \times \nv}$ be the \textit{Moore-Penrose pseudoinverse} of the weighted Laplacian matrix $L$. It is well-known that both $L$ and $L^+$ are real, symmetric, positive semi-definite matrices.
Denote by $\xx \in \R^{\ne}$ the vector of current flows across the network edges and by $\yy \in R^{\nv}$ the vector of the potentials in the network nodes. The currents flow in the network $(G,w)$ according to Ohm's law, which can be written in matrix form as
\begin{equation}
\label{eq:x}
	\xx =D A \yy
\end{equation}

Let $\ff \in\R^{\nv}$ be the vector describing the external currents: for every $i=1,\dots,\nv$, $\ff_i$ represents the current generated (if $\ff_i>0$) or consumed (if $\ff_i<0$) at node $i$. We assume that the energy network is balanced, which means that the total power generation equals the total power demand:
\[
	\mathbf{1}^T \ff = \sum_{i=1}^{\nv} \ff_i = 0.
\]
Such a zero-sum external current vector $\ff$ induces potentials $\yy$ in the network nodes given by
\begin{equation}
\label{eq:y}
	\yy = L^+ \ff,
\end{equation}
and, in view of~\eqref{eq:x}, a current vector
\begin{equation}
\label{eq:xf}
	\xx =D A L^+ \ff.
\end{equation}
In this paper, we will be mostly interested in the scenarios where $\ff$ is a random vector (or a multi-dimensional stochastic process), which is useful to model energy networks where both the power demand and generation are stochastic in nature.

\subsection{Total heat loss as quadratic form}
It is well know that when a current $x$ is flowing in a line with conductance $w$ dissipates some power is dissipated as heat at rate $\frac{1}{2} w^{-1} x^2$. Let $H =H(\ff)$ be the total heat loss on the energy network $G$ corresponding to the external current vector $\ff$, that is
\[
	H := \sum_{l \in E} \frac{1}{2} w_{l}^{-1} \xx_{l}^2 =\sum_{(j,k) \in E} \frac{w_{j,k}}{2}  (\yy_j-\yy_k)^2,
\]
where the vectors $\xx$ and $\yy$ depend on the external current vector $\ff$ as illustrated in~\eqref{eq:x} and~\eqref{eq:y}, respectively. This identity in matrix form reads
\begin{equation}
\label{eq:hdxy}
	H= \frac{1}{2} \xx^T D^{-1} \xx = \frac{1}{2} \yy^T L \yy.
\end{equation}

For our analysis, it is convenient to make explicit the dependence of the total heat loss on the external current vector $\ff$. The next proposition shows how $H$ can be written as a quadratic form of the external current vector $\ff$.

In particular, if $\ff$ is a random vector (or a stochastic process), then the total heat loss also becomes a random variable (or a stochastic process, respectively) and we show how to compute the expected total heat loss $h:=\E H$ dissipated by an energy network knowing the network structure and the first two moments of the distribution of the external currents.

\begin{prop}[Total heat loss and its expected value] \label{prop:qf}
Consider the network $(G,w)$ and let $L^+$ be the associated Laplacian pseudo-inverse. Given an external current vector $\ff$, the total heat loss in the network is
\begin{equation}\label{eq:hd}
	H =  \frac{1}{2} \, \ff^T L^+ \ff.
\end{equation}
If the external current vector $\ff$ is a random vector with mean $\bm{\mu}$ and covariance matrix $\bm{\Sigma}$, then the expected total heat loss in the network $(G,w)$ is
\begin{equation}\label{eq:Eh}
	\E H = \frac{1}{2} \, \tr{L^+ \bm{\Sigma}} + \frac{1}{2} \, \bm{\mu}^T L^+ \bm{\mu}.
\end{equation}
\end{prop}
\begin{proof}
Combining identities~\eqref{eq:LADA},~\eqref{eq:xf}, and~\eqref{eq:hdxy}, the total heat loss given an external current vector $\ff$ rewrites as
\begin{align*}
	\xx^T D^{-1} \xx
	&= \ff^T (D A L^+)^T D^{-1} D A L^+ \ff \\
	& =  \ff^T (L^+)^T A^T D D^{-1} D A L^+ \ff \\
	& =  \ff^T L^+ A^T D A L^+ \ff \\
	&=  \ff^T L^+ L L^+ \ff = \ff^T L^+ \ff,
\end{align*}
where in the last step of the latter equation we used the well-known identity $L^+ L L^+ = L^+$. If $\ff$ is a random vector, then $H$ becomes a random variable and its expected value can be computed using~\cite[Corollary 3.2b.1]{MP92}, obtaining~\eqref{eq:Eh}.
\end{proof}


\subsection{Effective resistance and Kirchhoff index}

The \textit{effective resistance} $R_{i,j}$ between a pair of nodes $i$ and $j$ is the potential difference that appears across nodes $i$ and $j$ when a unit current source is applied between them:
\[
	R_{i,j} := \yy_i - \yy_j, \quad \text{ with } \, \yy = L^+ (\mathbf{e}_i - \mathbf{e}_j),
\]
or, equivalently,
\[
	R_{i,j} := (\mathbf{e}_i - \mathbf{e}_j)^T L^+ (\mathbf{e}_i - \mathbf{e}_j).
\]
As shown in~\cite{KR93} the effective resistance defines a distance on the set of nodes. Indeed,
\begin{itemize}
\item $R_{i,i}=0$ for every $i=1,\dots,n$ and $R_{i,j} >0$ if $i \neq j$;
\item $R_{i,j}=R_{j,i}$ for every $i \neq j$;
\item For $i,j,k = 1,\dots,n$ the triangle inequality holds:
\begin{equation}
\label{eq:triangular}
	R_{i,j} \leq R_{i,k} + R_{k,j} \quad \forall \, k =1,\dots,n.
\end{equation}
\end{itemize}
For this reason, the effective resistance $R_{i,j}$ is also known in the literature as the \textit{resistance distance} between the nodes $i$ and $j$.
Let $R$ be the \textit{resistance matrix} of the network $(G,w)$, that is the $n \times n$ matrix whose $(i,j)$ entry is $R_{i,j}$. From the properties of the effective resistance, it immediately follows that $R$ is a symmetric matrix with all diagonal entries equal to zero.

%

Consider the network $(G,\mathbf{1})$, in which all the conductances are set equal to $1$.  The \textit{Kirchhoff index} or \textit{total effective resistance} of a graph $G$ is then defined using the resistance matrix $R$ as
\[
	Kf(G):=\frac{1}{2} \sum_{i,j=1}^n R_{i,j}.
\]
The Kirchhoff index was introduced in~\cite{KR93} and in the same paper the authors showed that 
\[
	Kf(G) = n \cdot \tr{L^+}.
\]
From this identity and~\eqref{eq:hd} it immediately follows that if the random external currents $\ff$ in the network $(G,\mathbf{1})$ are independent with zero mean and unit variance, then
\[
	\E H = \frac{Kf(G)}{2n}.
\]
Such an index has been extensively studied in the chemistry literature and it has been calculated for many graphs of interest~\cite{KLG95,LNT99,EJS76}. The Kirchhoff index has also been proposed in~\cite{ESVJK11} as a measure of graph robustness and in ~\cite{GBS08} the authors consider the problem of how to minimize such a quantity in a given graph.


In this work, we need to go beyond the notion of the Kirchhoff index for several reasons. Firstly, the quantity $Kf(G)$ describes the expected total heat loss in a energy network where the current injections are \textit{not} independent, since they need to satisfy the identity $\mathbf{1}^T \ff=0$. Furthermore, we are interested in the study of energy networks with heterogeneous conductances and where the external current injections could possibly have heterogeneous means and variances.


\section{One battery}
\label{sec3}
Consider an energy network with a single battery, which we assume to be placed at node $n$, without loss of generality. For every $i=1,\dots,n-1$ we define $\FF_i(t)$ as the random external current at node $i$ at time $t$. Let $\FF(t):=(\FF_1(t),\dots,\FF_{n-1}(t))$ be the random vector describing the external currents at time $t$ in the whole network and denote by $\bm{\mu}(t) \in \R^{\nv}$ its average and by $\bm{\Sigma}(t) \in \R^{(n-1) \times (n-1)}$ its covariance matrix at time $t$.

We will henceforth assume that the external currents in all $n-1$ nodes are independent, so that at time $t$ the covariance matrix $\bm{\Sigma}(t)$ is a diagonal matrix with diagonal entries that we denote as $\s_1^2(t),\dots,\s_{n-1}^2(t)$.

We model the battery as a \textit{slack node}, \ie we assume that the battery compensates for the total power generation-demand mismatch $\sum_{i=1}^{n-1} \FF_i(t)$, either by charging using the current in excess when $\sum_{i=1}^{n-1} \FF_i(t)>0$ or discharging to meet the demand when $\sum_{i=1}^{n-1} \FF_i(t)<0$. This means that the actual external current vector at time $t$ is the $n$-dimensional random vector $\ff(t)$ defined as
\[
	\ff_i(t) :=
	\begin{cases}
		 \FF_i(t) & \text{ if } i=1,\dots,n-1,\\
		 -\sum_{i=1}^{n-1} \FF_i(t) & \text{ if } i = n.
	\end{cases}
\]
Note that the relation between the vectors $\ff(t)$ and $\FF(t)$ can be rewritten in matrix form as
\begin{equation}
\label{eq:ffoneFF}
	\ff(t) =\prm \FF(t),
\end{equation}
where $\prm$ is the $n \times (n-1)$ matrix defined as
\[
	\prm =
	\left(
	\begin{array}{c}
		I_{n-1}\\
		- \mathbf{1}
	\end{array}
	\right),
\]
with $I_{n-1}$ being the identity matrix of order $n-1$ and $\mathbf{1} \in \R^{n-1}$ the vector with all entries equal to $1$.

The total heat loss at time $t$ for this energy network with one battery, which we denote by $H(t)$, is given by
\begin{align}
	H(t) &=  \frac{1}{2} \, \ff(t)^T L^+ \ff(t) = \frac{1}{2} \, \FF(t)^T \prm^T L^+ \prm \FF(t) \nonumber \\
	& \stackrel{\eqref{eq:ffoneFF}}{=} \frac{1}{2} \, \FF(t)^T \ls \FF(t), \label{eq:hds1}
\end{align}
where the matrix $\ls \in \R^{(n-1) \times (n-1)}$ is defined as
\[
		\ls:= \prm^T L^+  \prm.
\]
The matrix $\ls$ contains all the relevant information about the network structure and conductances as well as the battery location. In view of~\eqref{eq:hds1} and Proposition~\ref{prop:qf}, the expected total heat loss $h(t):= \E H(t)$ at time $t$ is
\[
	h(t) = \frac{1}{2} \,\tr{\bm{\Sigma}(t) \ls} + \frac{1}{2} \,\bm{\mu}(t)^T \ls \bm{\mu}(t).
\]
The next theorem, which is proved in the Appendix, illustrates how the expected total heat loss $h(t)$ at time $t$ depends on the network structure, on the battery location, and on the first two moments of the external current distribution.
\newpage
\begin{thm}[Expected total heat loss with one battery]%
\label{thm:onebattery}
Consider a network $(G,w)$ with associated resistance matrix $R$. If a battery is placed at node $n$, the expected total heat loss $h(t)$ at time $t$ is equal to
\begin{align}
	\hspace{-0.2cm} h(t) =&  \sum_{i=1}^{n-1} \frac{\s_i^2(t)}{2} R_{n,i} + \frac{1}{2} \Big (\sum_{i=1}^{n-1} \mu_i(t) \Big )\Big ( \sum_{i=1}^{n-1} \mu_i(t) R_{n,i} \Big ) \nonumber \\
	& - \frac{1}{2} \sum_{i,j=1,\, i<j}^{n-1} \mu_i(t) \mu_j(t) R_{i,j}. \label{eq:onebattery}
\end{align}
\end{thm}
A special case of interest is the situation where the expected generation-demand mismatch is equal to zero, \ie
\[
	\sum_{i=1}^{n-1} \mu_i(t)=0.
\]
In this scenario the energy network is on average self-balanced, hence the main role of the battery is to shove peaks in generation or demand, as well as to compensate for spatial imbalances in the network in the current production or consumption. In this case, the expected total heat loss simplifies to
\[
	h(t) =  \frac{1}{2} \sum_{i=1}^{n-1} \s_i^2(t) R_{n,i} - \frac{1}{2} \sum_{i,j=1,\, i<j}^{n-1} \mu_i(t) \mu_j(t) R_{i,j}.
\]
Both this formula and~\eqref{eq:onebattery} can be used to determine the optimal placement of a battery in a network. We shall apply this in the next section to a line network, investigating the difference between one and two optimally located batteries. 


\section{Two batteries}
\label{sec4}
We now consider the scenario where in the energy network there are multiple batteries that can coordinate their operations. We focus here the case of two batteries, which we assume to be placed in nodes $n-1$ and $n$, without loss of generality.

The remaining $n-2$ nodes of the network have a stochastic current injections or consumptions. Similarly to the previous section, we denote by $\FF(t)$, $\bm{\mu}(t)$, and $\bm{\Sigma}(t)$ the random external current vector at time $t$, its mean and its covariance matrix, respectively. Furthermore, also here we assume that the $n-2$ external currents are independent, which means that $\bm{\Sigma}(t)$ is a diagonal matrix with entries $\s_1^2(t),\dots,\s_{n-2}^2(t)$.

The two batteries compensate for the total power generation-demand mismatch $\sum_{i=1}^{n-2} \FF_i(t)$ according to the weights $\a$ and $1-\a$, respectively, for some $\a \in \R$ to which we will refer as \textit{load sharing coefficient}. This means that the external currents in the whole network at time $t$ are described by the $n$-dimensional vector $\ff(t)$ defined as
\[
	\ff_i(t) =
	\begin{cases}
		 \FF_i(t) & \text{ if } i=1,\dots,n-2,\\
		 -\a \sum_{i=1}^{n-2} \FF_i(t) & \text{ if } i = n-1,\\
		 -(1-\a) \sum_{i=1}^{n-2} \FF_i(t) & \text{ if } i = n.\\
	\end{cases}
\]
Note that when $\a \in [0,1]$, both batteries simultaneously charge or discharge (since the weights $\a$ and $1-\a$ have the same sign), while they have a specular behavior when $\a \in (-\infty,0) \cup (1,+\infty)$. Note that the relation between the vectors $\ff(t)$ and $\FF(t)$ can be rewritten in matrix form as
\begin{equation}
\label{eq:fftwoFF}
	\ff(t):=\prmt(\a) \FF(t),
\end{equation}
where $\prmt(\a)$ be the $n \times (n-2)$ matrix defined as
\[
	\prmt(\a) =
	\left(
	\begin{array}{c}
		I_{n-2}\\
		-\a \, \mathbf{1}\\
		-(1-\a) \, \mathbf{1}
	\end{array}
	\right).
\]
Let $H(t,\a)$ be the total heat loss in the network at time $t$ that corresponds to the external current vector $\FF(t)$ when the load sharing coefficient of the two batteries is $\a$. Using~\eqref{eq:hd} and~\eqref{eq:fftwoFF}, we can write
\begin{equation}
\label{eq:hds2}
	H(t,\a) = \frac{1}{2} \, \FF(t)^T \lst(\a) \FF(t), 
\end{equation}
where $\lst(\a):= \prmt^T(\a) L^+  \prmt(\a)$ is an $(n-2 )\times (n-2)$ matrix that encodes all the crucial information about the two batteries, their location, and the network properties. 

From~\eqref{eq:hds2} and Proposition~\ref{prop:qf} it follows that when the two batteries use a load sharing coefficient $\a$, the expected total heat loss $h(t,\a) := \E H(t,\a)$ at time $t$ is equal to
\[
	 h(t,\a) = \frac{1}{2} \,\tr{\bm{\Sigma}(t) \lst(\a)} + \frac{1}{2} \, \bm{\mu}(t)^T \lst(\a) \bm{\mu}(t).
\]
The next theorem shows how such a quantity depends on the network structure, on the first two moments of the external current distribution and on the load sharing coefficient $\a$.

\begin{thm}[Expected total heat loss with two batteries]
\label{thm:heat2batteries}
Consider a network $(G,w)$ with associated resistance matrix $R$. If two batteries are placed at nodes $n-1$ and $n$ and use a load sharing coefficient $\a$, the expected total heat loss $h(t,\a)$ at time $t$ is equal to
\begin{align}
	&h(t,\a)  \, = \, \frac{\a^2}{2} \, R_{n-1,n} \, \Big  ( \sum_{i=1}^{n-2} \s_i^2(t) + \Big  (\sum_{i=1}^{n-2} \mu_i(t) \Big )^2 \, \Big ) \nonumber \\
	&   + \frac{\a}{2} \Big (\sum_{i=1}^{n-2} \s_i^2(t) \D_i + \Big (\sum_{i=1}^{n-2} \mu_i(t) \Big ) \Big (\sum_{i=1}^{n-2} \mu_i(t) \D_i \Big ) \Big )  \nonumber\\
	&   + \frac{1}{2} \, \sum_{i=1}^{n-2} \s_i^2(t)  R_{n,i} + \frac{1}{2} \,\Big (\sum_{i=1}^{n-2} \mu_i(t) \Big ) \Big (\sum_{i=1}^{n-2} \mu_i(t) R_{n,i} \Big ) \nonumber \\
	& - \frac{1}{2} \,\sum_{i,j=1, \, i<j}^{n-2} \mu_i(t) \mu_j(t) R_{i,j}, \label{eq:heat2batteries}
\end{align}
where $\D_{i}:=R_{n-1,i} - R_{n,i}-R_{n-1,n} $, for $i=1,\dots,n-2$.
\end{thm}
The proof of this theorem is presented in the Appendix. 

In view of~\eqref{eq:heat2batteries}, $h(t,\a)$ is a second-degree convex polynomial in $\a$ and it is easy to show that $h(t,\a) >0$ for any $\a$ and any choice for the external current parameters. Furthermore, $h(t,\a)$ has a unique minimum as a function of $\a$, which we denote by $\a^*(t)$ and refer to as the optimal load sharing coefficient at time $t$. One can derive that
\[
	\a^*(t) = - \frac{\sum_{i=1}^{n-2} \s_i^2(t) \D_i + \sum_{j=1}^{n-2} \mu_j(t) \sum_{i=1}^{n-2} \mu_i(t) \D_i}{2 \, R_{n-1,n} \Big  (\sum_{i=1}^{n-2} \s_i^2(t) + \Big  (\sum_{i=1}^{n-2} \mu_i(t) \Big )^2 \, \Big )}.
\]
Such an optimal load sharing coefficient depends on the means and variances of the external currents, as well as on the network structure via the effective resistance $R_{n-1,n}$ between the two batteries and the quantities $\D_i$, $i=1,\dots,n-2$.

In the scenario where the network is on average self-balanced, \ie $\sum_{i=1}^{n-2} \mu_i(t)=0$, the optimal load sharing coefficient reads
\[
	\a^*(t) = - \frac{\sum_{i=1}^{n-2} \s_i^2(t) \D_i}{2 \, R_{n-1,n} \, \sum_{i=1}^{n-2} \s_i^2(t)}.
\]
It is easy to prove using~\eqref{eq:triangular} and the definition of $\D_i$ that in this special case the load sharing coefficient $\a^*(t)$ lies in $[0,1]$, so that at any time $t$ the two batteries either both charge or both discharge; note that this is not necessarily the case in the general case described above.



\subsection{Performance improvement: an example}
In this subsection we illustrate the performance gain that can be obtained by using two batteries instead of one. In order to obtain a more transparent result we consider a simple network topology, namely a line network of $n$ nodes with unit conductance edges, and i.i.d.~external currents with zero mean and variance $\sigma^2$.

One can show that the optimal location for a single battery is the central node $s^*=\lfloor \frac{n+1}{2} \rfloor$, while the optimal displacement of two batteries with load sharing coefficient $\a=1/2$ is
\[
	s_1^* =\left \lfloor \frac{n}{4} \right \rfloor, \qquad s_2^*= \left \lfloor \frac{3 n}{4} \right \rfloor +1.
\]
The corresponding expected total heat losses are
\[
	h_1(n) = \frac{n^2 \sigma ^2}{8}-\frac{\sigma ^2}{8}
\quad \text{ and } \quad
	h_2(n)  = \frac{3 n^2 \sigma ^2}{32}-\frac{n \sigma ^2}{8}-\frac{\sigma ^2}{4}.
\]
Hence, by adding a battery we reduce the expected total heat loss by $25\%$ asymptotically in the number of nodes, since
\[
	\frac{h_2(n)}{h_1(n)} \to \frac{3}{4}, \quad \text{ as } n \to \infty.
\]


\section{Two batteries: dynamic control}
\label{sec5}
In the previous section we derived the optimal control for two batteries as function of the network structure and external current parameters. However, in order to compute the optimal load sharing coefficient $\a^*(t)$ at time $t$ one would need to know all these data in real-time, which is not realistic. Furthermore, the time evolution of $\a^*(t)$ could be quite irregular, if the external current parameters changes too rapidly.

To tackle both these issues, we investigate how to obtain a smoother control on a finite time interval $[0,T]$. The time-scale assumption discussed in Section~\ref{sec1} allows us to take a stationary external current processes $\FF(t)$, so that the expected external current vector does not depend on time, \ie $\bm{\mu}(t)=\bm{\mu}$, and the covariance matrix is a diagonal matrix $\bm{\Sigma}$ and does not depend on the time either.

In view of Theorem~\ref{thm:heat2batteries}, the expected total heat loss $h(t)$ using a load sharing coefficient $\a(t)$ at time $t$ rewrites as
\begin{equation}
\label{eq:htabc}
	h(t) = a \, \a(t)^2 + b \, \a(t) + c,
\end{equation}
where $a$, $b$, and $c$ are three time-independent coefficients that depend on the network structure and on the mean and variances of the external current processes, see~\eqref{eq:heat2batteries}.

Let $\g >0$ be a positive real constant and assume that the load sharing coefficient $\a: [0,T] \to \R$ is two times differentiable on $[0,T]$. Consider the variational problem
\begin{equation}
\label{eq:varprob}
		\begin{cases}
	\displaystyle \min_{\a(t)} \int_0^T & a \, \a(s)^2 + b \, \a(s) +c + \gamma (\alpha'(s))^2 \, \mathrm{d}s, \\
	\textup{such that} &\a(0)=d \in \R,
	\end{cases}
\end{equation}
In other words, we would like to find the optimal dynamical load sharing coefficient that minimizes the expected total heat loss, knowing that at time $t=0$ the load sharing coefficient is $\a(0) =d$. The penalty term $\gamma (\alpha'(s))^2$ has been added to the variational problem to obtain an load sharing coefficient $\a(t)$ that evolves smoothly over time, without abrupt changes. 
\begin{thm}[Optimal dynamic control of two batteries]%
\label{thm:control}
The solution of the variational problem~\eqref{eq:varprob} is
\begin{equation}
\label{eq:control}
	\a_{\mathrm{opt}}(s) = -\frac{b}{2a} + \frac{d+\frac{b}{2a}}{ \mathrm{e}^{2 T \sqrt{\frac{a}{\gamma }}}+1} \left (\mathrm{e}^{s \sqrt{\frac{a}{\g}}} +\mathrm{e}^{(2T-s) \sqrt{\frac{a}{\g}}}\right ).
\end{equation}
\end{thm}
\begin{proof}
The Euler-Lagrange equation of problem~\eqref{eq:varprob} reads
\[
	2 a \cdot \alpha (s)+b-2 \gamma \cdot  \alpha ''(s)=0.
\]
A general solution of this non-homogeneous second-order linear ordinary differential equation reads
\begin{equation}
\label{eq:solution}
	\a(s) = - \frac{b}{2a} + c_1 \mathrm{e}^{- s \sqrt{\frac{a}{\g}}} + c_2 \mathrm{e}^{ s \sqrt{\frac{a}{\g}}}.
\end{equation}
Assume that the derivative at zero is equal to $\a'(0)=e \in \R$. This, together with the initial conditions $\a(0)=d$, yields
\begin{align}
	c_1(e) &= \frac{b+2 a d-2 e \sqrt{a \gamma }}{4 a}, \nonumber \\
	c_2(e) &= \frac{ b + 2 a d + 2 e \sqrt{a \gamma }}{4 a}. \label{eq:c1c2}
\end{align}
We now minimize over the value of $\a'(0)=e$. Define
\[
	M(e) := \int_0^T a \cdot \a(s)^2 + b \cdot \a(s) +c + \gamma (\a'(s))^2 \, \mathrm{d}s.
\]
From~\eqref{eq:solution} and~\eqref{eq:c1c2} it follows that the function $M(e)$ is a second degree convex polynomial in $e$ and attains its minimum at
\begin{equation}
\label{eq:estar}
	e^* = -\frac{(2 a d+b) \left(\mathrm{e}^{2 T \sqrt{\frac{a}{\gamma }}}-1\right)}{2 \sqrt{a \gamma } \left(\mathrm{e}^{2 T \sqrt{\frac{a}{\gamma }}}+1\right)}.
\end{equation}
Combining~\eqref{eq:solution},~\eqref{eq:c1c2} and~\eqref{eq:estar}, we obtain that the solution of the variational problem is~\eqref{eq:control}.
\end{proof}

\subsection{An example: Ornstein-Uhlenbeck external currents}
In this subsection we will briefly illustrate the scenario where the external currents are modeled by independent stationary Ornstein-Ulhenbeck processes. More specifically, we assume that for every $ i=1,\dots, n-2$ the external current at node $i$ is a Ornstein–Uhlenbeck process, with mean $\mu_i$, variance $\s_i^2$ and mean-reversion rate $\theta_i$, \ie
\[
	d \FF_i(t) = \theta_i (\mu_i - \FF_i(t) ) d t + \s_i^2 \, d W_i(t),
\]
where $(W_i(t))_{i=1}^{n-2}$ are independent standard Brownian motions. 
As illustrated in~\eqref{eq:htabc} the expected total heat loss $h(t)$ in the network when the two batteries use a load sharing coefficient $\a(t)$ at time $t$ can be written as
\[
	h(t) = a \, \a(t)^2 + b \, \a(t) + c,
\]
where, in view of the assumption for the external currents to evolve as Ornstein–Uhlenbeck processes, the coefficients $a,b,c$ can be calculated as
\begin{align*}
	a&=\frac{1}{2} \Big (\sum_{i=1}^{n-2} \frac{\s_i^2}{2\theta_i} +\Big (\sum_{j=1}^{n-2}  \mu_j \Big )^2 \, \Big ) R_{n-1,n}, \\
	b&=\frac{1}{2} \,  \sum_{i=1}^{n-2} \Big (\frac{\s_i^2}{2\theta_i} + \mu_i \Big (\sum_{j=1}^{n-2}  \mu_j \Big ) \Big ) \D_i,\\
	c&= \frac{1}{2} \, \sum_{i=1}^{n-2} \Big(\frac{\s_i^2}{2\theta_i} + \mu_i \Big (\sum_{j=1}^{n-2}  \mu_j \Big ) \Big ) R_{n,i} \\
	& \quad - \frac{1}{2} \,\sum_{i,j=1, \, i<j}^{n-2} \mu_i \mu_j R_{i,j}.
\end{align*}
As an example, we consider a network $(G,\mathbf{1})$ where the graph $G$, illustrated in Figure~\ref{fig:network}, is the IEEE 14-bus network, which has $n=14$ nodes and $m=20$ edges. The two batteries have been placed in the positions highlighted in black (first battery) and white (second battery).
\begin{figure}[!h]
	\centering
	\includegraphics[scale=0.32]{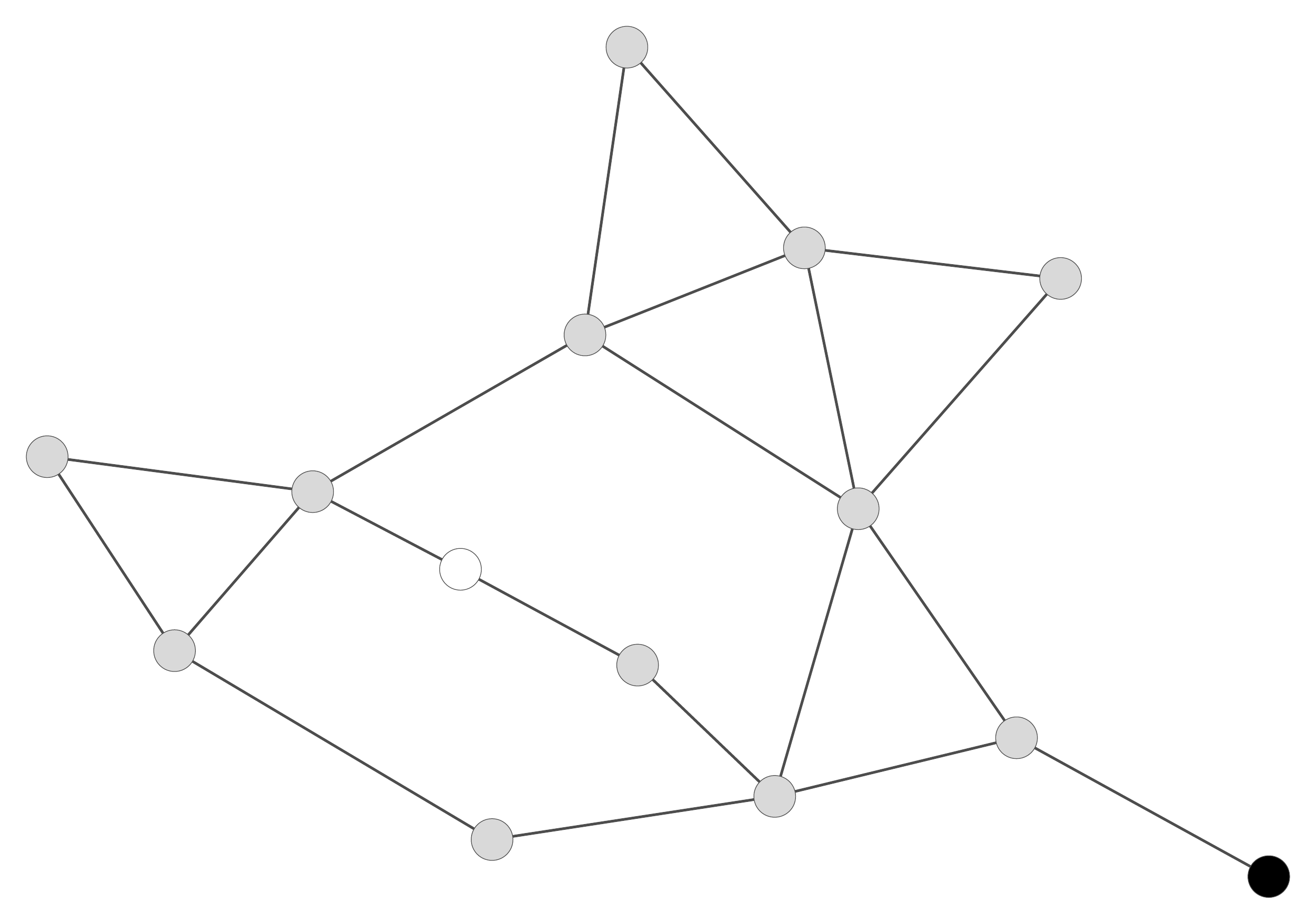}
	\caption{The graph $G$ describing the IEEE 14-bus network topology}
	\label{fig:network}
\end{figure}

We simulate the external currents as independent Ornstein–Uhlenbeck processes with heterogeneous parameters to have qualitative insight of the performance of the optimal smooth control for the load sharing coefficient derived in Theorem~\ref{thm:control}. 

Figure~\ref{fig:controls} shows the evolution over time of $\a_{\mathrm{opt}}(t)$ and that of the pathwise omniscient optimal control $\a_{\mathrm{omn}}(t)$, \ie the one obtained by minimizing the total heat loss at every time \textit{conditionally on the realized external currents}.

Despite knowing only the mean, variance and mean-reversion rate of the external current processes, the smooth control $\a_{\mathrm{opt}}(t)$ performs quite well and indeed the total heat loss tracks quite closely the minimum total heat loss, achievable with the omniscient control, as illustrated in Figure~\ref{fig:heatpath}.
\begin{figure}[!h]
	\centering
	\includegraphics[scale=0.6]{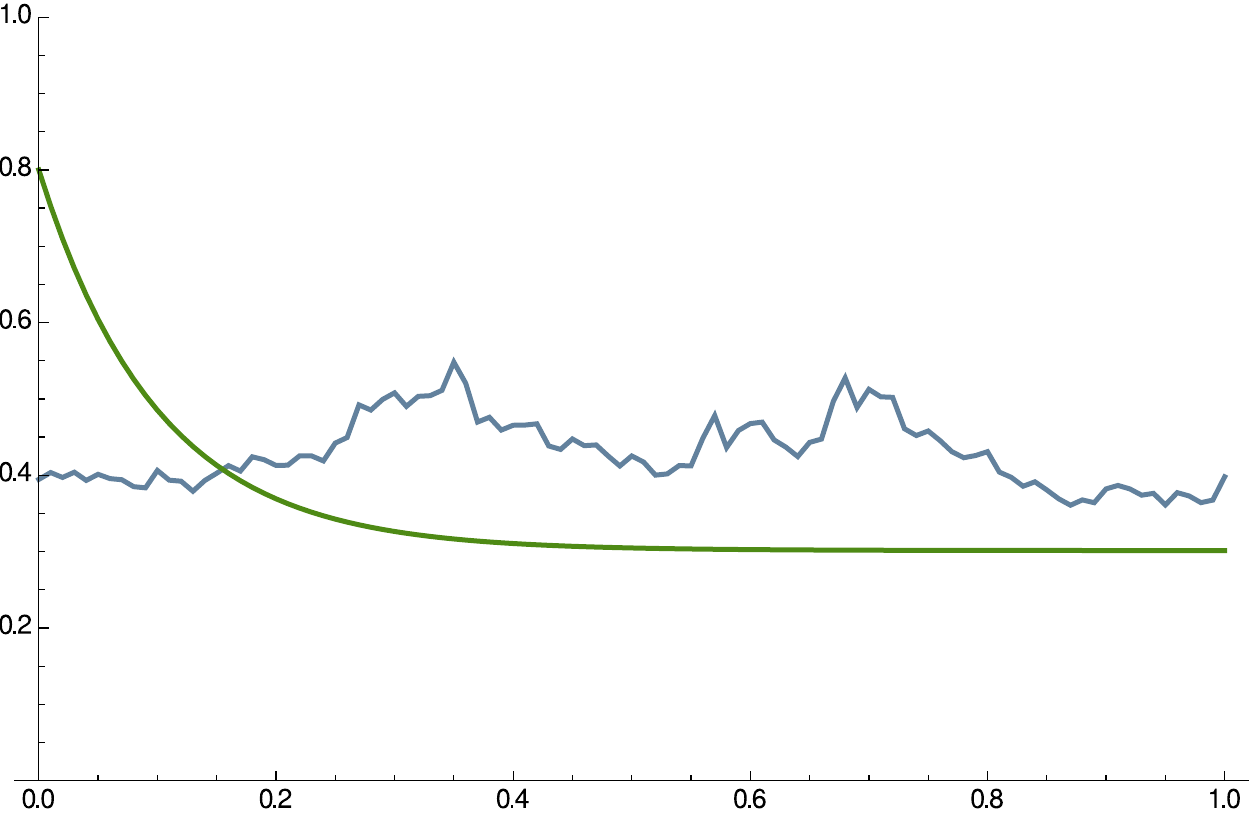}
	\caption{The smooth control $\a_{\mathrm{opt}}(t)$ (in green) vs. the path-wise omniscent optimal control $\a_{\mathrm{omn}}(t)$ (in blue)}
	\label{fig:controls}
\end{figure}
\begin{figure}[!h]
	\centering
	\includegraphics[scale=0.6]{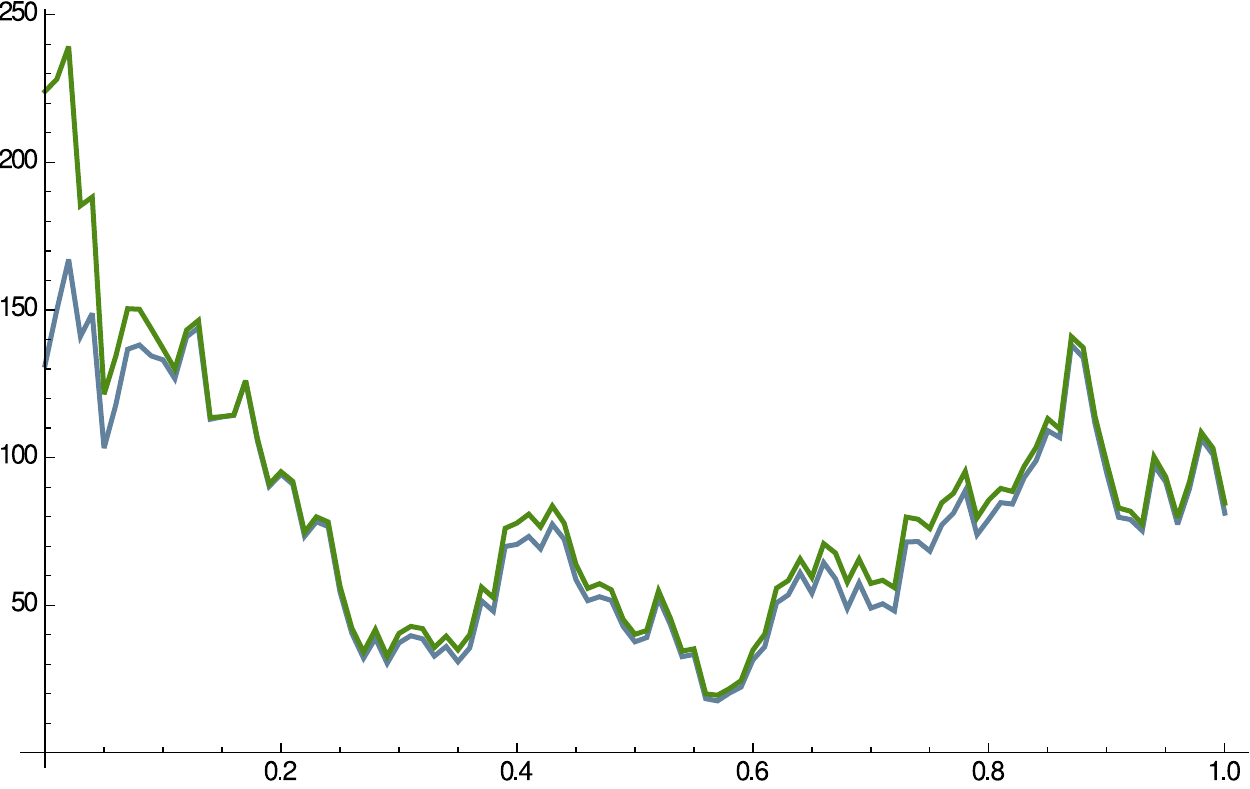}
	\caption{The total heat loss evolutions corresponding to the two controls visualized in Figure~\ref{fig:controls}}
	\label{fig:heatpath}
\end{figure}


\FloatBarrier
\section{Conclusions}
\label{sec6}
In this paper we consider a stochastic model for energy networks, which we model as weighted graphs with random sources and sinks. We analyzed the impact of storage devices using the total heat loss as performance metric and explicit results have been obtained in the case of one and two batteries, this latter also in a dynamic context. We leave the extension to scenarios with an arbitrary number of batteries for future work.


\section{Appendix}
\begin{lem}[Matrix $L^+$ in terms of effective resistances {\cite[Theorem 4.8]{BCEG09}}]\label{lem:lp}
\begin{align}
	\hspace{-0.2cm} L^+_{i,j} &= - \frac{1}{2} \Big ( R_{i,j} - \frac{1}{n} \sum_{k=1}^n (R_{i,k} + R_{j,k}) + \frac{1}{n^2} \sum_{k,l=1}^n R_{k,l} \Big ) \nonumber\\
	&=  \frac{1}{2n} \sum_{k=1}^n (R_{i,k} + R_{j,k}-R_{i,j}) - \frac{1}{n^2} Kf(G). \label{eq:lp}
\end{align}
\end{lem}

\begin{lem}[Matrix $\ls$ in terms of effective resistances] \label{lem:L1}
\[
	(\ls)_{i,j} = \frac{1}{2}(R_{n,i}+R_{n,j} - R_{i,j}).
\]
\end{lem}
\begin{proof}
Denote by $\mathbf{J}$ the a matrix with all entries equal to $1$ and by $\mathbf{0}$ that with all entries equal to $0$. Since every column of $\prm$ has zero sum, it follows that
\[
	\prm^T \mathbf{J} \prm = \mathbf{0}.
\]
Thus, if we denote $B:=L^+ +\frac{1}{n^2} Kf(G)\mathbf{J}$, then $B$ is a symmetric matrix and
\[
	\prm^T L^+  \prm  = \prm^T B \prm.
\]
Furthermore, in view of~\eqref{eq:lp},
\[
	B_{i,j} =  \frac{1}{2n} \sum_{k=1}^n (R_{i,k} + R_{j,k}-R_{i,j}),
\]
and thus
\begin{align*}
	& (\ls)_{i,j} \nonumber = \text{row}_i (\prm^T(\a)) \, B \, \text{column}_j(\prm(\a)) \nonumber \\
	& \quad  = B_{i,j} - B_{n,i} - B_{n,j} + B_{n,n}  \nonumber \\
	& \quad =  \frac{1}{2n} \sum_{k=1}^n (R_{j,k}-R_{i,j} - R_{n,k} + R_{n,i})  \nonumber \\
	& \qquad + \frac{1}{2n} \sum_{k=1}^n ( -R_{j,k} - R_{n,k} + R_{n,j}+2 R_{n,k}) \nonumber \\
	& \quad =  \frac{1}{2n} \sum_{k=1}^n (R_{n,i} + R_{n,j} - R_{i,j}) \nonumber \\
	& \quad =  \frac{1}{2} (R_{n,i} + R_{n,j} - R_{i,j}). \nonumber  \qquad \qquad \qquad \qquad  \qquad \qedhere
\end{align*}
\end{proof}

\begin{proof}[Proof of Theorem~\ref{thm:onebattery}]
For notational compactness, we suppress in this proof the dependence on $t$ of the various parameters. From Lemma~\ref{lem:L1}, it immediately follows that $(\ls)_{i,i}=R_{n,i}$, that, together with the assumption that $\bm{\Sigma}=\mathrm{diag}(\s^2_1,\dots, \s^2_{n-1})$, yields%
\[
	\tr{\bm{\Sigma} \ls}  = \sum_{i=1}^{n-1} \s_i^2 R_{n,i}.
\]
Thanks to Lemma~\ref{lem:L1} and the properties of the matrix $R$ 
\begin{align*}
	\hspace{-0.1cm}\bm{\mu}^T \ls \bm{\mu} &= \frac{1}{2} \sum_{i=1}^{n-1}  \sum_{j=1}^{n-1} \mu_i \mu_j (R_{n,i}+R_{n,j} - R_{i,j})\\
	&= \frac{1}{2} \sum_{i=1}^{n-1}  \sum_{j=1}^{n-1} \mu_i \mu_j (R_{n,i}+R_{n,j})\\
	& \quad - \frac{1}{2} \sum_{i=1}^{n-1}  \sum_{j=1}^{n-1} \mu_i \mu_j R_{i,j}\\
	&= \sum_{i=1}^{n-1}  \sum_{j=1}^{n-1} \mu_i \mu_j R_{n,j} - \sum_{i=1}^{n-1}  \sum_{j=i+1}^{n-1} \mu_i \mu_j R_{i,j}\\
	&= \Big( \sum_{i=1}^{n-1} \mu_i  \Big) \Big( \sum_{j=1}^{n-1} \mu_j R_{n,j} \Big) - \sum_{\substack{i,j=1\\i <j}}^{n-1} \mu_i \mu_j R_{i,j},
\end{align*}
where the second last passage follows from the identity 
\begin{equation}
	\sum_{i=1}^{n-1}  \sum_{j=1}^{n-1} b_i b_j (a_i+a_j) = 2 \sum_{i=1}^{n-1}  \sum_{j=1}^{n-1} b_i b_j a_j
	\label{eq:funsum}
\end{equation}
valid for any $a_1,\dots,a_{n-1}$ and $b_1,\dots,b_{n-1}$ in $\R$.
\end{proof}

\begin{lem}[Matrix $\lst(\a)$ in terms of effective resistances] \label{lem:L2}
\begin{align*}
	\hspace{-0.1cm} (\lst(\a))_{i,j} & = -\a(1-\a) R_{n-1,n} + \a \frac{R_{n-1,i} + R_{n-1,j}}{2} \\
	& \quad + (1-\a) \frac{R_{n,i} + R_{n,j}}{2} - \frac{R_{i,j}}{2}.
\end{align*}
\end{lem}
\begin{proof}
For every $\a \in \R$, $\prmt^T(\a)\mathbf{1}=1-\a-(1-\a)=0$ and, hence,
\[
	\prmt^T(\a) \mathbf{J} \prmt(\a) = \mathbf{0}.
\]
Hence, we can argue like in the proof of Lemma~\ref{lem:L1} and consider the symmetric matrix $B:=L^+ +\frac{1}{n^2} Kf(G)\mathbf{J}$, which is such that
\[
	\prmt^T(\a) L^+  \prmt(\a)  = \prmt^T(\a) B \prmt(\a).
\]
From~\eqref{eq:lp} it follows that
\begin{align}
\label{eq:Aij}
	B_{i,j} & =\frac{1}{2n} \sum_{k=1}^n (R_{i,k} + R_{k,j}-R_{i,j}) \nonumber \\
	& = -\frac{R_{i,j}}{2} + \frac{1}{2n} \sum_{k=1}^n (R_{i,k} + R_{k,j}).
\end{align}
Using the auxiliary matrix $B$, we can write
\begin{align}
	& (\lst(\a))_{i,j} \nonumber = \text{row}_i (\prmt^T(\a)) \, B \, \text{column}_j(\prmt(\a)) \nonumber \\
	& \quad = B_{i,j} - \a B_{i,n-1} - (1-\a) B_{i,n}  \nonumber \\
	& \qquad - \a \left (B_{n-1,j} - \a B_{n-1,n-1} - (1-\a) B_{n-1,n}\right ) \nonumber \\
	& \qquad -(1-\a)\left  (B_{n,j} - \a B_{n,n-1} - (1-\a) B_{n,n}\right )
	\label{eq:terms}
\end{align}
Denote $S_{n-1}:= \sum_{i=1}^{n-2} R_{i,n-1}$ and $S_{n}:= \sum_{i=1}^{n-2} R_{i,n}$. Consider first the terms in~\eqref{eq:terms} that do not depend on $i$ and $j$: using identity~\eqref{eq:Aij} we get
\begin{align}
	\a^2 &B_{n-1,n-1} + (1-\a)^2 B_{n,n}+2\a(1-\a) B_{n-1,n} =\nonumber \\
	&= -\a(1-\a) R_{n-1,n} + \frac{\a(1-\a)}{n} (S_{n-1} + S_{n}) \nonumber\\
	& \quad + \frac{\a^2}{n} S_{n-1} + \frac{(1-\a)^2}{n} S_{n},\nonumber\\
	&= -\a(1-\a) R_{n-1,n} +\frac{1}{n}\left ( \a S_{n-1} + (1-\a) S_{n}  \right ), \label{eq:firstterms}
\end{align}
Consider now the terms in~\eqref{eq:terms} that depend on $i$ and $j$. By virtue of~\eqref{eq:Aij}, we derive
\begin{align}
	& B_{i,j} - \a (B_{i,n-1} + B_{n-1,j}) - (1-\a) (B_{i,n} + B_{n,j}) =\nonumber \\
	& \quad \, = \a \frac{R_{n-1,i} + R_{n-1,j}}{2} +(1-\a) \frac{R_{n,i} + R_{n,j}}{2} \nonumber \\
	& \qquad \, - \frac{1}{n}\left (\a S_{n-1} + (1-\a)  S_{n} \right )  -\frac{R_{i,j}}{2}. \label{eq:secondterms}
\end{align}
The proof is completed by combining~\eqref{eq:terms}-\eqref{eq:secondterms}.
\end{proof}

\begin{proof}[Proof of Theorem~\ref{thm:heat2batteries}]
Also in this proof, we suppress the dependence on $t$ for notational compactness. Since $\bm{\Sigma}=\mathrm{diag}(\s^2_1,\dots, \s^2_{n-2})$, using Lemma~\ref{lem:L2} and the properties of matrix $R$ we get that 
\begin{align*}
	\tr{\lst(\a)}  &= \sum_{i=1}^{n-2} \s_i^2 ( \a^2 R_{n-1,n}  + \a  \D_i + R_{n,i}),
\end{align*}
with $\D_i=R_{n-1,i} -R_{n,i}-R_{n-1,n}$ for $i=1,\dots,n-2$. Moreover, again by virtue of Lemma~\ref{lem:L2},
\begin{align*}
	\hspace{-0.15cm} \bm{\mu}^T& \lst(\a) \bm{\mu}= \sum_{i=1}^{n-2} \sum_{j=1}^{n-2} \mu_i \mu_j (\lst(\a))_{i,j}\\
	& \, =  \a^2 \sum_{i=1}^{n-2} \sum_{j=1}^{n-2} \mu_i \mu_j R_{n-1,n} \\
	& \quad + \frac{\a}{2} \sum_{i=1}^{n-2} \sum_{j=1}^{n-2} \mu_i \mu_j (\D_i +\D_j)\\
	& \quad + \frac{1}{2} \sum_{i=1}^{n-2} \sum_{j=1}^{n-2} \mu_i \mu_j (R_{n,i}+R_{n,j})\\
	& \quad - \frac{1}{2} \sum_{i=1}^{n-2} \sum_{j=1}^{n-2} \mu_i \mu_j R_{i,j}\\
	& \, = \a^2 R_{n-1,n} \Big (\sum_{i=1}^{n-2} \mu_i \Big )^2 +\\
	& \quad  +\a \Big (\sum_{i=1}^{n-2} \mu_i(t) \Big ) \Big (\sum_{i=1}^{n-2} \mu_i(t) \D_i \Big ) \\
	& \quad  + \Big (\sum_{i=1}^{n-2} \mu_i(t) \Big ) \Big (\sum_{i=1}^{n-2} \mu_i(t) R_{n,i} \Big ) - \sum_{\substack{i,j=1\\i<j}}^{n-2} \mu_i \mu_j R_{i,j}.
\end{align*}
where in the last step we used identity~\eqref{eq:funsum}.
\end{proof}



\end{document}